\documentclass{llncs}

\usepackage{amsmath,amssymb}
\usepackage{graphicx}

\usepackage{algs2}
\newcommand{\C}{\mathcal{C}}
\newcommand{\SSS}{\mathcal{S}}

\newcommand{\comment}[1]{}

\title{Algorithms for the Majority Rule (+) Consensus Tree and
the Frequency Difference Consensus Tree}

\author{
  Jesper Jansson\inst{1}
\and
  Chuanqi Shen\inst{2}
\and
  Wing-Kin Sung\inst{3,4}
}

\institute{
  Laboratory of Mathematical Bioinformatics (Akutsu Laboratory), \\
  Institute for Chemical Research, \\
  Kyoto University, Gokasho, Uji, Kyoto~611-0011, Japan. \\
  E-mail:~\email{jj@kuicr.kyoto-u.ac.jp} \\
  Funded by The Hakubi Project 
  and KAKENHI grant number~23700011.
\and
  Stanford University, 450~Serra Mall, Stanford, CA~94305-2004, U.S.A. \\
  E-mail:~\email{shencq@stanford.edu}
\and
  School of Computing,
  National University of Singapore,
  13~Computing Drive, Singapore~117417. \\
  E-mail: \texttt{ksung@comp.nus.edu.sg}
\and
  Genome Institute of Singapore,
  60~Biopolis Street, Genome, Singapore~138672.
}

\date{}

\begin{document}

\maketitle

\begin{abstract}
This paper presents two new deterministic algorithms for constructing
consensus trees.
Given an input
%set
of $k$~phylogenetic trees with identical leaf label sets
and $n$~leaves each, the first algorithm constructs
the \emph{majority rule~(+) consensus tree} in $O(k n)$ time, which is
optimal since the input size is $\Omega(k n)$,
and the second one constructs the \emph{frequency difference consensus tree}
in $\min \{O(k n^{2}),\, O(k n (k + \log^{2}n))\}$ time.
\end{abstract}

%%%%%%%%%%%%%%%%%%%%%%%%%%%%%%%%%%%%%%%%%%%%%%%%%%%%%%%%%%%%%%%%%%%%%%%%%%%%%%

\section{Introduction}
\label{section: Introduction}

A \emph{consensus tree} is a phylogenetic tree that summarizes a given
collection of phylogenetic trees having the same leaf labels but different
branching structures.
Consensus trees are used to resolve structural differences between two or
more existing phylogenetic trees arising from conflicts in the raw data,
to find strongly supported groupings, and to summarize large sets of
candidate trees obtained by bootstrapping when trying to infer a new
phylogenetic tree
%accurately~\cite{ACS03,BDF-B_11,DDBR09,book:Fel04,KWY98,book:Sung_10}.
accurately~\cite{ACS03,DDBR09,book:Fel04,book:Sung_10}.

Since the first type of consensus tree was proposed by
Adams~III~\cite{A72} in~1972, many others have been defined and analyzed.
See, e.g.,~\cite{chapter:Bryant03}, Chapter~30 in~\cite{book:Fel04},
or Chapter~8.4 in~\cite{book:Sung_10} for some surveys.
Which particular type of consensus tree to use in practice depends on the
context.
For example, the \emph{strict consensus tree}~\cite{SR81} is very intuitive
and easy to compute~\cite{D85} and may be sufficient when there is not so
much disagreement in the data,
the \emph{majority rule consensus tree}~\cite{MM81} is ``the optimal tree to
report if we view the cost of reporting an estimate of the phylogeny to be
a linear function of the number of incorrect clades in the estimate and
the number of true clades that are missing from the estimate and we view 
the reporting of an incorrect grouping as a more serious error than missing
a clade''~\cite{HSL08},
and the \emph{R* consensus tree}~\cite{chapter:Bryant03} provides
a statistically consistent estimator of the species tree topology when
combining gene trees~\cite{DDBR09}.
%So that scientists can use whichever type of consensus tree that best fits
%the application at hand or compare the outputs of various consensus methods
%to better understand variations in the data,
%efficient algorithms for constructing a broad range of different consensus
%trees are needed.
Therefore, scientists need
efficient algorithms for constructing a broad range of different consensus
trees.

In a recent series of papers~\cite{CJS_12,JSS_13a,JSS_13b,JS_13}, we have
developed fast algorithms for computing
the \emph{majority rule consensus tree}~\cite{MM81},
the \emph{loose consensus tree}~\cite{B90}
(also known in the literature as
the \emph{combinable component consensus tree} or
the \emph{semi-strict consensus tree}),
a \emph{greedy consensus tree}~\cite{chapter:Bryant03,software:Fel05},
the \emph{R* consensus tree}~\cite{chapter:Bryant03},
and consensus trees for so-called
\emph{multi-labeled phylogenetic trees (MUL-trees)}~\cite{LSHPOM09}.
In this paper, we study two relatively new types of consensus trees called
the \emph{majority rule~(+) consensus tree}~\cite{CW07,DF-BMP_10}
and the \emph{frequency difference consensus tree}~\cite{GFKORS03},
and give algorithms for constructing them efficiently.

\subsection{Definitions and notation}

\begin{figure}[t!]
%\hspace*{14.5mm}
\hspace*{1mm}
\begin{tabular}{ccccccccc}
  \includegraphics[scale=0.35]{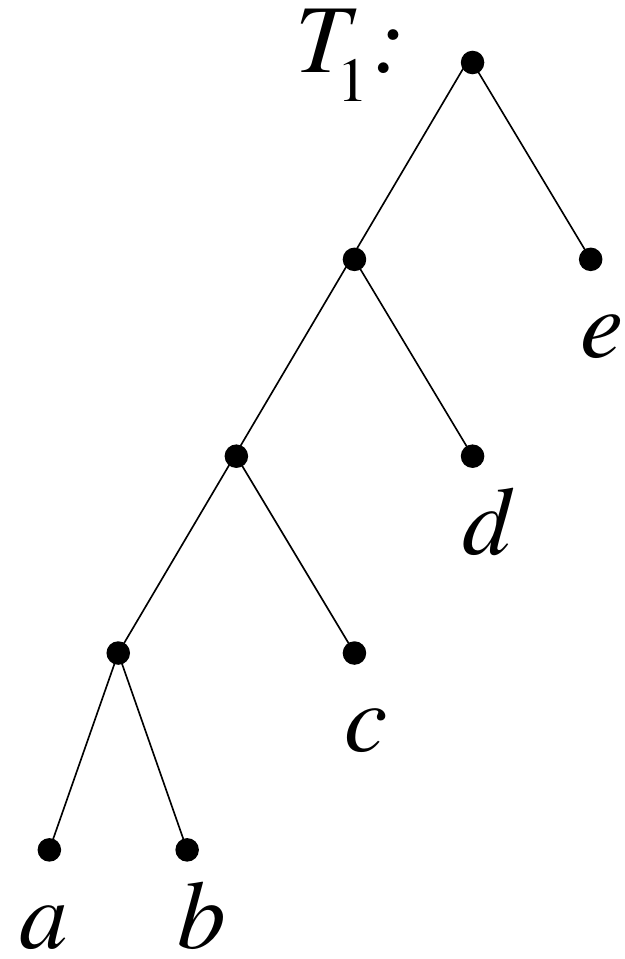}
& \hspace*{5mm} &
  \includegraphics[scale=0.35]{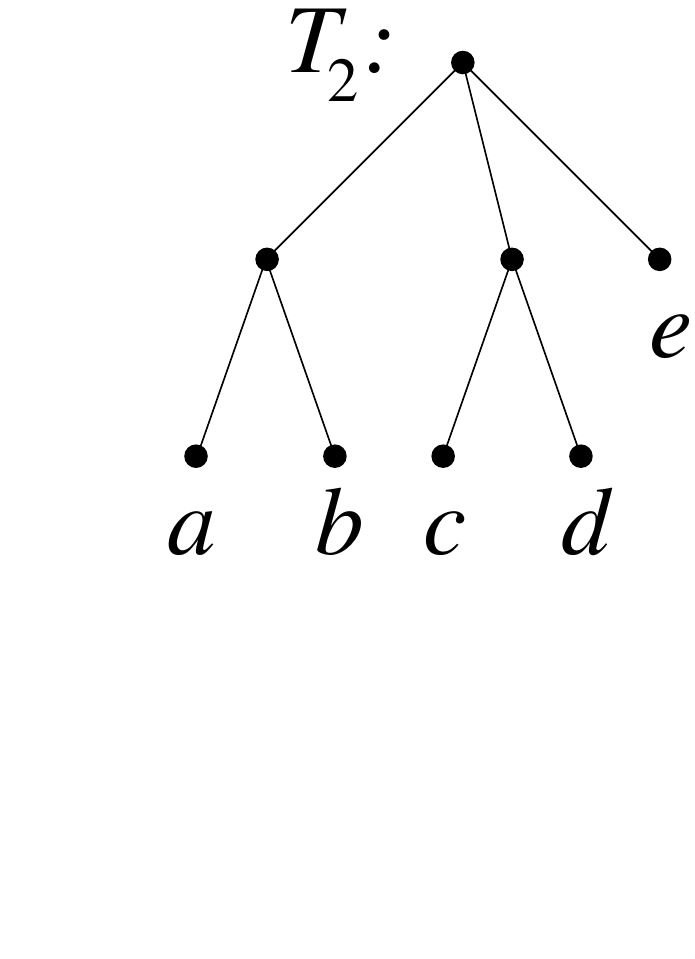}
& \hspace*{5mm} &
  \includegraphics[scale=0.35]{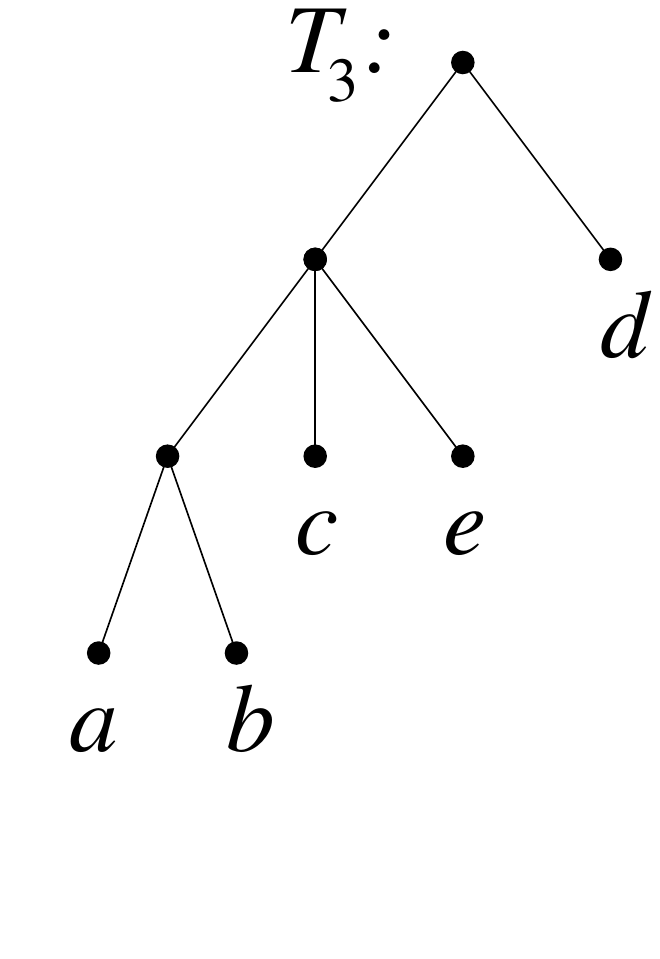}
& \hspace*{5mm} &
  \includegraphics[scale=0.35]{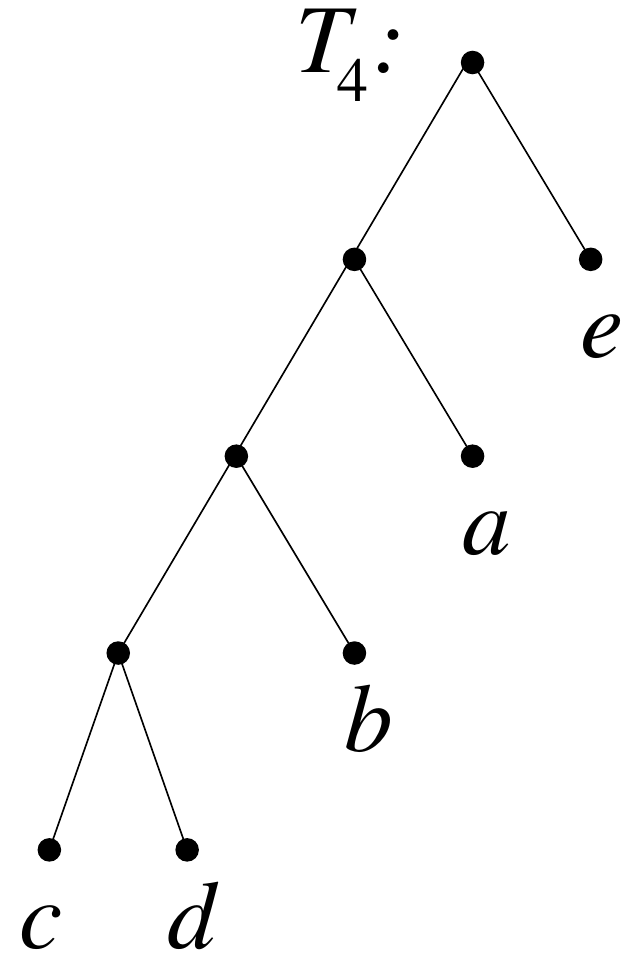}
\\[3mm]
  \includegraphics[scale=0.35]{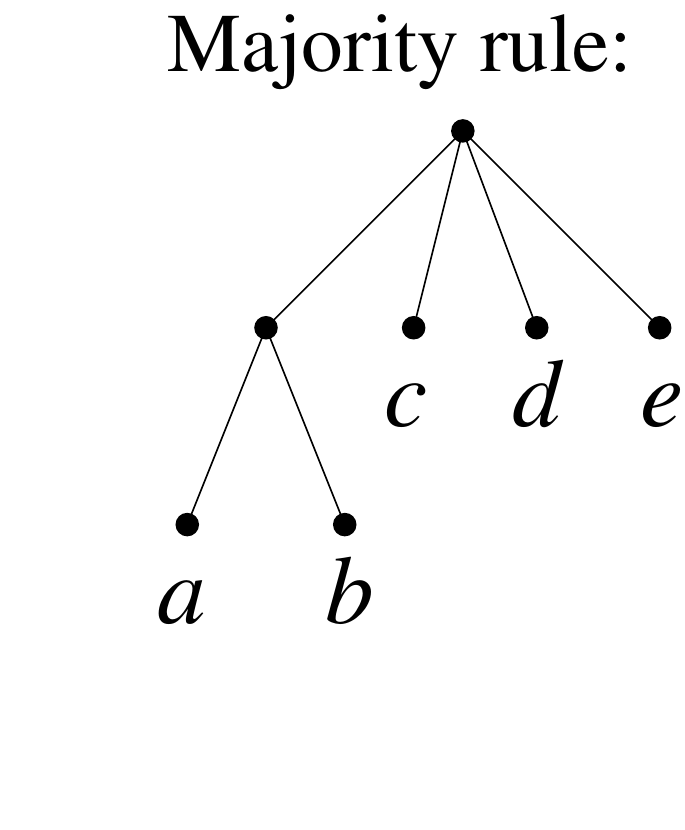}
& \hspace*{5mm} &
  \includegraphics[scale=0.35]{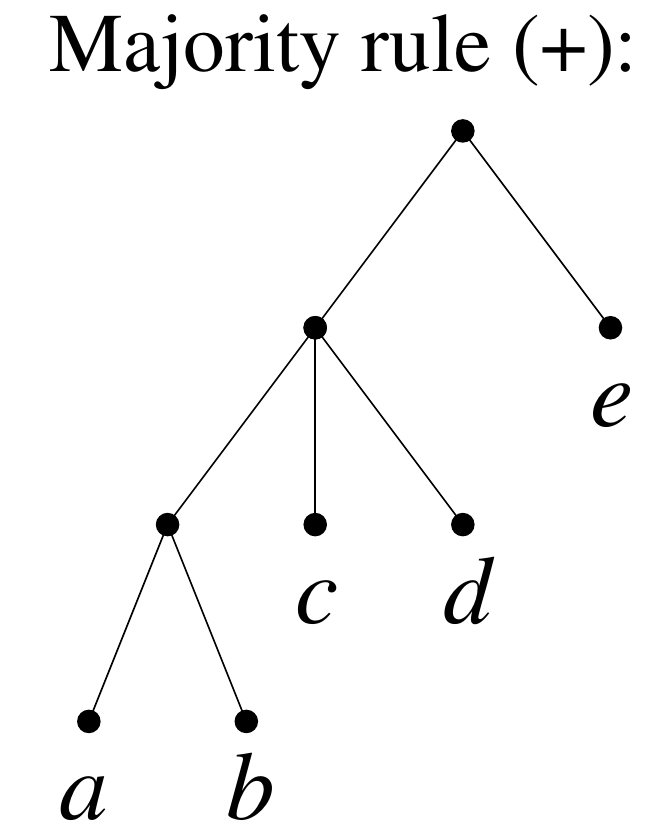}
& \hspace*{5mm} &
  \includegraphics[scale=0.35]{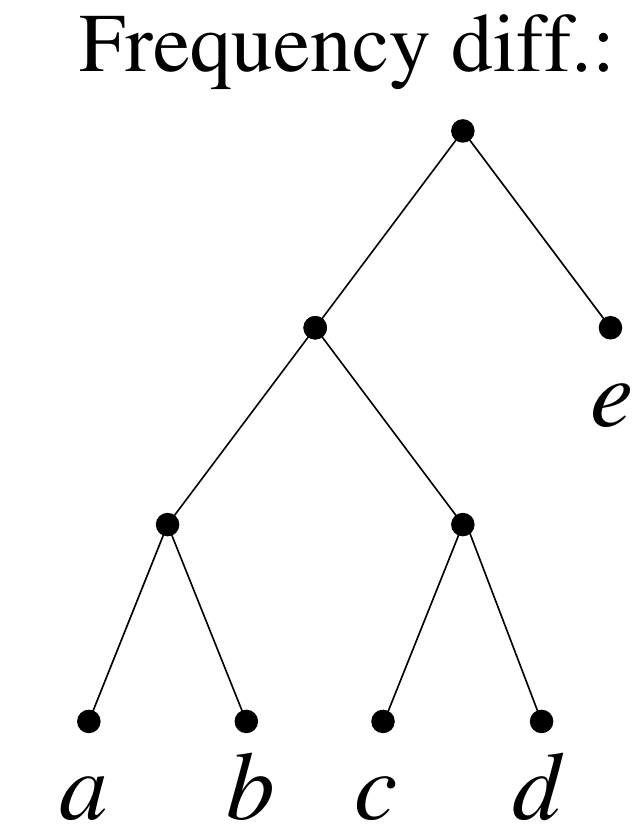}
\end{tabular}
\caption{Let $\SSS = \{T_1, T_2, T_3, T_4\}$ as shown above with
$L = \Lambda(T_1) = \Lambda(T_2) = \Lambda(T_3) = \Lambda(T_4) =
  \{a,b,c,d,e\}$.
%\newline
The only non-trivial majority cluster of~$\SSS$ is
$\{a,b\}$,
the non-trivial majority~(+) clusters of~$\SSS$ are
$\{a,b\}$ and $\{a,b,c,d\}$, and
the non-trivial frequency difference clusters of~$\SSS$ are
$\{a,b\}$, $\{a,b,c,d\}$, and $\{c,d\}$.
The majority rule, majority rule~(+), and frequency difference consensus
trees of~$\SSS$ are displayed.
}
\label{figure: frequency_difference_example}
\end{figure}

We shall use the following basic definitions.
A \emph{phylogenetic tree} is a rooted, unordered, leaf-labeled tree in which
every internal node has at least two children and all leaves have different
labels.
(Below, phylogenetic trees are referred to as ``trees'' for short).
For any tree~$T$, the set of all nodes in~$T$ is denoted by~$V(T)$ and
the set of all leaf labels in~$T$ by~$\Lambda(T)$.
Any nonempty subset~$C$ of~$\Lambda(T)$ is called a \emph{cluster}
of~$\Lambda(T)$;
if $|C| = 1$ or $C = \Lambda(T)$ then $C$ is \emph{trivial}, and otherwise,
$C$~is \emph{non-trivial}.
For any $u \in V(T)$, $T[u]$~denotes the subtree of~$T$ rooted at
the node~$u$.
Observe that if $u$ is the root of~$T$ or if $u$ is a leaf then
$\Lambda(T[u])$ is a trivial cluster.
The set $\C(T) \,=\, \bigcup\nolimits_{u \in V(T)} \{\Lambda(T[u])\}$
is called the \emph{cluster collection of~$T$}, and
any cluster $C \subseteq \Lambda(T)$ is said to \emph{occur in~$T$}
if $C \in \C(T)$.

Two clusters~$C_1, C_2 \subseteq \Lambda(T)$ are \emph{compatible}
if $C_1 \subseteq C_2$, \,$C_2 \subseteq C_1$, or $C_1 \cap C_2 = \emptyset$.
If $C_1$ and $C_2$ are compatible, we write $C_1 \smile C_2$;
otherwise, $C_1 \not\smile C_2$.
A cluster $C \subseteq \Lambda(T)$ is \emph{compatible with~$T$} if
$C \smile \Lambda(T[u])$ holds for every node $u \in V(T)$.
In this case, we write $C \smile T$, and $C \not\smile T$ otherwise.
If $T_1$ and~$T_2$ are two trees with $\Lambda(T_1) = \Lambda(T_2)$ such that
every cluster in~$\C(T_1)$ is compatible with~$T_2$ then it follows that
every cluster in~$\C(T_2)$ is compatible with~$T_1$, and we say that
$T_1$ and~$T_2$ are \emph{compatible}.
Any two clusters or trees that are not compatible are called
\emph{incompatible}.

Next, let $\SSS \,=\, \{T_1, T_2, \dots, T_k\}$ be a set of trees satisfying
$\Lambda(T_1) = \Lambda(T_2) = \dots = \Lambda(T_k) = L$ for some leaf label
set~$L$.
For any cluster~$C$ of~$L$, denote the set of all trees in~$\SSS$ in which
$C$ occurs by~$K_{C}(\SSS)$ and the set of all trees in~$\SSS$ that are
incompatible with~$C$ by~$Q_{C}(\SSS)$.
Thus,
$K_{C}(\SSS) = \{T_i \,:\, C \in \C(T_i)\}$
and $Q_{C}(\SSS) = \{T_i \,:\, C \not\smile T_i\}$.
Define three special types of clusters:
\begin{itemize}
\item[{\raise0.9pt\hbox{$\bullet$}}]
If $|K_{C}(\SSS)| > \frac{k}{2}$
then $C$ is a \emph{majority cluster of~$\SSS$}.
\medskip
\item[{\raise0.9pt\hbox{$\bullet$}}]
If $|K_{C}(\SSS)| > |Q_{C}(\SSS)|$
then $C$ is a \emph{majority~(+) cluster of~$\SSS$}.
\medskip
\item[{\raise0.9pt\hbox{$\bullet$}}]
If $|K_{C}(\SSS)| >
  \max \{|K_{D}(\SSS)| \,:\, D \subseteq L \textnormal{ and } C \not\smile D\}$
then $C$ is a \emph{frequency difference cluster of~$\SSS$}.
\end{itemize}
(Informally, a frequency difference cluster is a cluster that occurs more
frequently than each of the clusters that is incompatible with it.)
Note that a majority cluster of~$\SSS$ is always a majority~(+) cluster
of~$\SSS$ and that a majority~(+) cluster of~$\SSS$ is always a frequency
difference cluster of~$\SSS$, but not the other way around.

The \emph{majority rule consensus tree of~$\SSS$}~\cite{MM81}
is the tree~$T$ such that $\Lambda(T) = L$ and
$\C(T)$ consists of all majority clusters of~$\SSS$.
Similarly,
the \emph{majority rule~(+) consensus tree of~$\SSS$}~\cite{CW07,DF-BMP_10}
is the tree~$T$ such that $\Lambda(T) = L$ and
$\C(T)$ consists of all majority~(+) clusters of~$\SSS$,
and
the \emph{frequency difference consensus tree of~$\SSS$}~\cite{GFKORS03}
is the tree~$T$ such that $\Lambda(T) = L$ and
$\C(T)$ consists of all frequency difference clusters of~$\SSS$.
See Fig.~\ref{figure: frequency_difference_example} for some examples.

From here on, $\SSS$~is assumed to be an input set of identically leaf-labeled
trees, and the leaf label set of~$\SSS$ is denoted by~$L$.
To express the size of the input, we define $k = |\SSS|$ and $n = |L|$.

\subsection{Previous work}
\label{subsection: Previous_work}

Margush and McMorris~\cite{MM81} introduced the majority rule consensus tree
in~1981, and a deterministic algorithm for constructing it in optimal
$O(k n)$ worst-case running time was presented recently in~\cite{JSS_13b}.
(A randomized algorithm with $O(k n)$ expected running time and unbounded
worst-case running time was given earlier by
Amenta \emph{et al.}~\cite{ACS03}.)
The majority rule consensus tree has several desirable mathematical
properties~\cite{BM86,HSL08,MP08},
and algorithms for constructing it have been implemented in popular
computational phylogenetics packages like
PHYLIP~\cite{software:Fel05},
TNT~\cite{software:GFN08},
COMPONENT~\cite{software:Pag93},
MrBayes~\cite{software:RH03},
SumTrees in DendroPy~\cite{software:SH_10},
and PAUP*~\cite{software:Swo03}.
Consequently, it is one of the most widely used consensus trees in
practice~\cite[p.~450]{CW07}.
One drawback of the majority rule consensus tree is that it may be too harsh
and discard valuable branching information.
For example, in Fig.~\ref{figure: frequency_difference_example}, even
though the cluster $\{a,b,c,d\}$ is compatible with $75\%$ of the input
trees, it is not included in the majority rule consensus tree.
For this reason, people have become interested in alternative types of
consensus trees that include all the majority clusters and at the same time,
also include other meaningful, well-defined kinds of clusters.
The majority rule~(+) consensus tree and the frequency difference consensus
tree are two such consensus trees.
%\footnote{In
%case all trees in~$\SSS$ are binary then
%the majority rule and the majority rule~(+) consensus trees are equal
%(Remark~1 in~\cite{DF-BMP_10}),
%but the frequency difference consensus tree may be more resolved
%(Fig.~4 in~\cite{DF-BMP_10}).}

The majority rule~(+) consensus tree was defined by
Dong~\emph{et al.}~\cite{DF-BMP_10} in~2010.
It was obtained as a special case of an attempted generalization by
Cotton and Wilkinson~\cite{CW07} of the majority rule consensus tree.
According to~\cite{DF-BMP_10}, Cotton and Wilkinson~\cite{CW07} suggested
two types of supertrees\footnote{A
\emph{supertree} is a generalization of a consensus tree that does not
require the input trees to have identical leaf label sets.}
called majority-rule~(-) and majority-rule~(+) that were supposed to
generalize the majority rule consensus tree.
Unexpectedly, only the first one did, and by restricting the second one to
the consensus tree case, \cite{DF-BMP_10}~arrived at the majority rule~(+)
consensus tree.
Dong \emph{et al.}~\cite{DF-BMP_10} established some fundamental properties
of this type of consensus tree and pointed out the existence of
a polynomial-time algorithm for constructing it, but left the task of finding
the best possible such algorithm as an open problem.
As far as we know, no implementation for computing the majority rule~(+)
consensus tree is publicly available.

Goloboff \emph{et al.}~\cite{GFKORS03} initially proposed the frequency
difference consensus tree as a way to improve methods for evaluating
group support in parsimony analysis.
Its relationships to other consensus trees have been studied
in~\cite{DF-BMP_10}.
A method for constructing it has been implemented in the free software
package TNT~\cite{software:GFN08} but the algorithm used is not documented
and its time complexity is unknown.
We note that since the number of clusters occurring in~$\SSS$ may be
$\Omega(k n)$, a naive algorithm that compares every cluster in~$\SSS$
to every other cluster in~$\SSS$ directly
would require~$\Omega(k^{2} n^{2})$~time.

\subsection{Organization of the paper and new results}

Due to space constraints, some proofs have been omitted from the conference
version of this paper.
Please see the journal version for the complete proofs.

The paper is organized as follows.
Section~\ref{section: Preliminaries} summarizes some results from
the literature that are needed later.
In Section~\ref{section: Majority_rule_plus_consensus_tree}, we modify
the techniques from~\cite{JSS_13b} to obtain an $O(k n)$-time algorithm for
the majority rule~(+) consensus tree.
Its running time is optimal because the size of the input is $\Omega(k n)$;
hence, we resolve the open problem of Dong \emph{et al.}~\cite{DF-BMP_10}
mentioned above.
Next, Section~\ref{section: Frequency_difference_consensus_tree} gives
a $\min \{O(k n^{2}),\, O(k n (k + \log^{2}n))\}$-time algorithm for
constructing the frequency difference consensus tree
(here, the second term is smaller than the first term if $k = o(n)$; e.g.,
if $k = O(1)$ then the running time reduces to $O(n \log^{2}n)$).
Our algorithms are fully deterministic and do not need to use hashing.
Finally, Section~\ref{section: Implementations} discusses implementations.

%%%%%%%%%%%%%%%%%%%%%%%%%%%%%%%%%%%%%%%%%%%%%%%%%%%%%%%%%%%%%%%%%%%%%%%%%%%%%%

\section{Preliminaries}
\label{section: Preliminaries}

\subsection{The $\mathit{delete}$ and $\mathit{insert}$ operations}

The $\mathit{delete}$ and $\mathit{insert}$ operations are two operations
that modify the structure of a tree.
They are defined in the following way.

Let $T$ be a tree and let $u$ be any non-root, internal node in~$T$.
The $\mathit{delete}$ operation on~$u$ makes all of $u$'s children become
children of the parent of~$u$, and then removes $u$ and the edge between~$u$
and its parent.
%(See Fig.~\ref{figure: delete_operation}.)
(See, e.g., Figure~2 in~\cite{JSS_13a} for an illustration.)
The time needed for this operation is proportional to the number of children
of~$u$, and the effect of applying it is that the cluster collection of~$T$
is changed to $\C(T) \setminus \{\Lambda(T[u])\}$.
Conversely, the $\mathit{insert}$ operation creates a new node~$u$ that
becomes:
(1)~a child of an existing internal node~$v$, and
(2)~the parent of a proper subset~$X$ of $v$'s children satisfying
$|X| \geq 2$;
the effect is that $\C(T)$ is changed to $\C(T) \cup \{\Lambda(T[u])\}$,
where $\Lambda(T[u]) = \bigcup_{v_i \in X} \Lambda(T[v_i])$.

%\begin{figure}[h!]
%\begin{center}
%  \includegraphics[scale=0.35]{fig_delete_operation_in_tree2-eps-converted-to.pdf}
%\caption{Figure from~\cite{JSS_13a}.
%Applying the $\mathit{delete}$ operation on~$u$ removes
%the cluster~$\{d,e,f\}$ from the cluster collection.}
%\label{figure: delete_operation}
%\end{center}
%\end{figure}

%%%

\subsection{Subroutines}
\label{subsection: Subroutines}

The new algorithms in this paper use the following algorithms from
the literature as subroutines:
Day's algorithm~\cite{D85},
Procedure \textnormal{\texttt{One-Way\_Compatible}}~\cite{JSS_13a},
and Procedure~\textnormal{\texttt{Merge\_Trees}}~\cite{JSS_13a}.
Day's algorithm~\cite{D85} is used to efficiently check whether any specified
cluster that occurs in a tree~$T$ also occurs in another
tree~$T_{\mathit{ref}}$, and can be applied to find the set of all clusters
that occur in both~$T$ and~$T_{\mathit{ref}}$ in linear time.
Procedure~\texttt{One-Way\_Compatible} takes as input two trees~$T_A$
and~$T_B$ with identical leaf label sets and outputs a copy of~$T_A$
in which every cluster that is not compatible with~$T_B$ has been removed.
(The procedure is asymmetric; e.g., if $T_A$ consists of $n$~leaves attached
to a root node and $T_B \neq T_A$ then
\texttt{One-Way\_Compatible}$(T_A, T_B)$ $= T_A$, while
\texttt{One-Way\_Compatible}$(T_B, T_A)$ $= T_B$.)
Procedure~\textnormal{\texttt{Merge\_Trees}} takes as input two compatible
trees with identical leaf label sets and outputs a tree that combines their
cluster collections.
Their properties are summarized below;
for details, see references~\cite{D85} and~\cite{JSS_13a}.

\begin{lemma}
\label{lemma: Day's_algorithm}
(Day~\cite{D85})
Let $T_{\mathit{ref}}$ and~$T$ be two given trees with
$\Lambda(T_{\mathit{ref}}) = \Lambda(T) = L$
and let $n = |L|$.
After $O(n)$ time preprocessing, it is possible to determine,
for any $u \in V(T)$, if $\Lambda(T[u]) \in \C(T_{\mathit{ref}})$ in
$O(1)$ time.
\end{lemma}

\begin{lemma}
\label{lemma: Procedure_One-Way_Compatible}
(\cite{JSS_13a})
Let $T_A$ and~$T_B$ be two given trees with
$\Lambda(T_A) = \Lambda(T_B) = L$
and let $n = |L|$.
Procedure \textnormal{\texttt{One-Way\_Compatible}}$(T_A, T_B)$
returns a tree~$T$ with $\Lambda(T) = L$ such that
$\C(T) =
  \{C \in \C(T_A) :
  \textnormal{$C$ is compatible with $T_B$}\}$
in $O(n)$ time.
\end{lemma}

\begin{lemma}
\label{lemma: Procedure_Merge_Trees}
(\cite{JSS_13a})
Let $T_A$ and~$T_B$ be two given trees with
$\Lambda(T_A) = \Lambda(T_B) = L$
that are compatible
and let $n = |L|$.
Procedure~\textnormal{\texttt{Merge\_Trees}}$(T_A, T_B)$
returns a tree~$T$ with $\Lambda(T) = L$ and $\C(T) = \C(T_A) \cup \C(T_B)$
in $O(n)$ time.
\end{lemma}

%%%%%%%%%%%%%%%%%%%%%%%%%%%%%%%%%%%%%%%%%%%%%%%%%%%%%%%%%%%%%%%%%%%%%%%%%%%%%%

\section{Constructing the majority rule~(+) consensus tree}
\label{section: Majority_rule_plus_consensus_tree}

This section presents an algorithm named \texttt{Maj\_Rule\_Plus} for
computing the majority rule~(+) consensus tree of~$\SSS$ in
(optimal) $O(k n)$ time.

The pseudocode of \texttt{Maj\_Rule\_Plus} is given in
Fig.~\ref{figure: Algorithm_Maj_Rule_Plus}.
The algorithm has two phases.
Phase~1 examines the input trees, one by one, to construct a set of candidate
clusters that includes all majority~(+) clusters.
Then, Phase~2 removes all candidate clusters that are not
majority~(+)~clusters.\footnote{This
basic strategy was previously used in the $O(k n)$-time algorithm
in~\cite{JSS_13b} for computing the majority rule consensus tree.
\comment{ %%% START OF COMMENT
The main difference is how the counters are updated in Phase~1; instead of
producing a superset of the majority clusters as in~\cite{JSS_13b}, we now
produce a superset of the majority~(+) clusters.
} %%% END OF COMMENT
}

During Phase~1, the current candidate clusters are stored as nodes in
a tree~$T$.
Every node~$v$ in~$T$ represents a current candidate cluster~$\Lambda(T[v])$
and has a counter $\mathit{count}(v)$ that,
starting from the iteration at which~$\Lambda(T[v])$ became a candidate
cluster, keeps track of the number of input trees in which it occurs minus
the number of input trees that are incompatible with it.
More precisely, while treating the tree~$T_j$ for any
$j \in \{2, 3, \dots, k\}$ in Step~\ref{substep: MRP_Phase_1_update_counters},
$\mathit{count}(v)$ for each current candidate cluster~$\Lambda(T[v])$
is updated as follows:
if $\Lambda(T[v])$ occurs in~$T_j$ then $\mathit{count}(v)$ is incremented
by~$1$,
if $\Lambda(T[v])$ does not occur in~$T_j$ and is not compatible with~$T_j$
then $\mathit{count}(v)$ is decremented by~$1$,
and otherwise (i.e., $\Lambda(T[v])$ does not occur in~$T_j$ but is
compatible with~$T_j$) $\mathit{count}(v)$ is unchanged.
Furthermore, if any $\mathit{count}(v)$ reaches~$0$ then the node~$v$ is
deleted from~$T$ so that $\Lambda(T[v])$ is no longer a current candidate
cluster.
Next, in Step~\ref{substep: MRP_Phase_1_insert_compatible_clusters}, every
cluster occurring in~$T_j$ that is not a current candidate but compatible
with~$T$ is inserted into~$T$ (thus becoming a current candidate cluster)
and its counter is initialized to~$1$.
Lemma~\ref{lemma: MRP_Phase_1_correctness} below proves that the set of
majority~(+) clusters of~$\SSS$ is contained in the set of candidate clusters
at the end of Phase~1.

\begin{figure}[t!]
\begin{center}
\fbox{
\begin{algorithm2}{\textnormal{\texttt{Maj\_Rule\_Plus}}}
\smallskip
\INPUT{
  A set $\SSS \,=\, \{T_1, T_2, \dots, T_k\}$ of trees with
  $\Lambda(T_1) = \Lambda(T_2) = \dots = \Lambda(T_k)$.
}
\OUTPUT{
  The majority rule (+) consensus tree of~$\SSS$.
}
\smallskip

\nullSTEP{
  /* Phase 1 */
}

\STEP{
  $T := T_1$
}

\STEP{
  \textbf{for} each $v \in V(T)$ \textbf{do}
  $\mathit{count}(v) := 1$
}

\STEP{
\label{step: MRP_Phase_1_main_loop}
  \textbf{for} $j := 2$ \textbf{to} $k$ \textbf{do}
  \subSTEP{\hspace*{5mm}
\label{substep: MRP_Phase_1_update_counters}
  \textbf{for} each $v \in V(T)$ \textbf{do}
  }
    \nullSTEP{\hspace*{10mm}
    \textbf{if} $\Lambda(T[v])$ occurs in~$T_j$ \textbf{then}
    $\mathit{count}(v) := \mathit{count}(v) + 1$
    }
    \nullSTEP{\hspace*{10mm}
    \textbf{else if} $\Lambda(T[v])$~is not compatible with~$T_j$ \textbf{then}
      $\mathit{count}(v) := \mathit{count}(v) - 1$
    }
  \nullSTEP{\hspace*{5mm}
  \textbf{endfor}
  }
  \subSTEP{\hspace*{5mm}
\label{substep: MRP_Phase_1_delete_loop}
  \textbf{for} each $v \in V(T)$ in top-down order \textbf{do}
  }
    \nullSTEP{\hspace*{10mm}
    \textbf{if} $\mathit{count}(v) = 0$ \textbf{then}
    $\mathit{delete}$ node~$v$.
    }
  \subSTEP{\hspace*{5mm}
\label{substep: MRP_Phase_1_insert_compatible_clusters}
  \textbf{for} every $C \in \C(T_j)$ that is compatible with~$T$ but
  does not occur in~$T$ \textbf{do}
  }
    \nullSTEP{\hspace*{10mm}
    Insert $C$ into~$T$.
    }
    \nullSTEP{\hspace*{10mm}
    Initialize $\mathit{count}(v) := 1$ for the new node~$v$ satisfying
    $\Lambda(T[v]) = C$.
    }
  \nullSTEP{\hspace*{5mm}
  \textbf{endfor}
  }
}
\nullSTEP{\textbf{endfor}}

\medskip

\nullSTEP{
  /* Phase 2 */
}

\STEP{
  \textbf{for} each $v \in V(T)$ \textbf{do}
  $\mathit{K}(v) := 0$;\, $\mathit{Q}(v) := 0$
}

\STEP{
\label{step: MRP_Phase_2_main_loop}
  \textbf{for} $j := 1$ \textbf{to} $k$ \textbf{do}
  \subSTEP{\hspace*{5mm}
  \textbf{for} each $v \in V(T)$ \textbf{do}
\label{substep: MRP_Phase_2_count_occurrences}
  }
  \nullSTEP{\hspace*{10mm}
  \textbf{if} $\Lambda(T[v])$ occurs in~$T_j$ \textbf{then}
    $\mathit{K}(v) := \mathit{K}(v) + 1$  
  }
  \nullSTEP{\hspace*{10mm}
  \textbf{else if} $\Lambda(T[v])$~is not compatible with~$T_j$ \textbf{then}
    $\mathit{Q}(v) := \mathit{Q}(v) + 1$  
  }
  \nullSTEP{\hspace*{5mm}
  \textbf{endfor}
  }
}

\STEP{
\label{step: MRP_Phase_2_delete_loop}
  \textbf{for} each $v \in V(T)$ in top-down order \textbf{do}
  \nullSTEP{\hspace*{5mm}
  \textbf{if} $\mathit{K}(v) \leq \mathit{Q}(v)$ \textbf{then}
    perform a $\mathit{delete}$ operation on~$v$.
  }
}

\STEP{
  \RETURN{$T$}
}
\end{algorithm2}
}
\caption{Algorithm~\texttt{Maj\_Rule\_Plus}
for constructing the majority rule~(+) consensus tree.}
\label{figure: Algorithm_Maj_Rule_Plus}
\end{center}
\end{figure}

\begin{lemma}
\label{lemma: MRP_Phase_1_correctness}
For any $C \subseteq L$, if $C$ is a majority~(+) cluster of~$\SSS$
then $C \in \C(T)$ at the end of Phase~1.
\end{lemma}

\begin{proof}
Suppose that $C$ is a majority~(+) cluster of~$\SSS$.
Let $T_x$ be any tree in~$Q_{C}(\SSS)$ and consider iteration~$x$ in
Step~\ref{step: MRP_Phase_1_main_loop}:
If $C$ is a current candidate at the beginning of iteration~$x$ then its
counter will be decremented, cancelling out the occurrence of~$C$ in
one tree~$T_j$ where $1 \leq j < x$;
otherwise, $C$~may be prevented from being inserted into~$T$ in at most one
later iteration~$j$ (where $x < j \leq k$ and $C \in \C(T_j)$)
because of some cluster occurring in~$T_x$.
It follows from $|K_{C}(\SSS)| - |Q_{C}(\SSS)| > 0$ that $C$'s counter will
be greater than~$0$ at the end of Phase~1, and therefore $C \in \C(T)$.
\qed
\end{proof}

In Phase~2, Step~\ref{step: MRP_Phase_2_main_loop} of the algorithm computes
the values of $|K_{C}(\SSS)|$ and $|Q_{C}(\SSS)|$ for every candidate
cluster~$C$ and stores them in $\mathit{K}(v)$ and $\mathit{Q}(v)$,
respectively, where $C = \Lambda(T[v])$.
Finally, Step~\ref{step: MRP_Phase_2_delete_loop} removes every candidate
cluster~$C$ that does not satisfy the condition
$|K_{C}(\SSS)| > |Q_{C}(\SSS)|$.
By definition, the clusters that remain in~$T$ are the majority~(+) clusters.

\begin{theorem}
Algorithm~\textnormal{\texttt{Maj\_Rule\_Plus}}
constructs the majority rule~(+) consensus tree of~$\SSS$ in $O(k n)$ time.
\end{theorem}

\begin{proof}
The correctness follows from Lemma~\ref{lemma: MRP_Phase_1_correctness} and
the above discussion.

%\medskip

The time complexity analysis is analogous to the proof of Theorem~4
in~\cite{JSS_13b}.
First consider Phase~1.
Step~\ref{substep: MRP_Phase_1_update_counters} takes $O(n)$ time by:
(1)~running Day's algorithm with $T_{\mathit{ref}} = T_j$ and then checking
each node~$v$ in~$V(T)$ to see if $\Lambda(T[v])$ occurs in~$T_j$
(according to Lemma~\ref{lemma: Day's_algorithm}, this requires $O(n)$
time for preprocessing, and each of the $O(n)$ nodes in~$V(T)$ may be checked
in $O(1)$ time), and
(2)~computing $X := \texttt{One-Way\_Compatible}(T, T_j)$ and then
checking for each node~$v$ in~$V(T)$ if $v$ does not exist in~$X$ to
determine if $\Lambda(T[v]) \not\smile T_j$
(this takes $O(n)$ time by
Lemma~\ref{lemma: Procedure_One-Way_Compatible}).
The $\mathit{delete}$ operations in Step~\ref{substep: MRP_Phase_1_delete_loop}
take $O(n)$ time because the nodes are handled in top-down order, which
means that for every node, its parent will change at most once in each
iteration.
In Step~\ref{substep: MRP_Phase_1_insert_compatible_clusters}, define
$Y := \texttt{One-Way\_Compatible}(T_j, T)$ and
$Z := \texttt{Merge\_Trees}(Y, T)$.
Then by Lemmas~\ref{lemma: Procedure_One-Way_Compatible}
and~\ref{lemma: Procedure_Merge_Trees}, the cluster collection of~$Y$
consists of the clusters occurring in~$T_j$ that are compatible with the set
of current candidates, and $Z$~is the result of inserting these clusters
into~$T$.
Thus, Step~\ref{substep: MRP_Phase_1_insert_compatible_clusters} can be
implemented by computing $Y$ and~$Z$, updating $T$'s structure according
to~$Z$, and setting the counters of all new nodes to~$1$, so
Step~\ref{substep: MRP_Phase_1_insert_compatible_clusters} takes $O(n)$ time.
The main loop in Step~\ref{step: MRP_Phase_1_main_loop} consists of
$O(k)$ iterations, and Phase~1 therefore takes $O(k n)$ time in total.

Next, Phase~2 also takes $O(k n)$ time because
Step~\ref{substep: MRP_Phase_2_count_occurrences} can be implemented in
$O(n)$ time with the same techniques as in
Step~\ref{substep: MRP_Phase_1_update_counters}, and
Step~\ref{step: MRP_Phase_2_delete_loop} is performed in $O(n)$ time by
handling the nodes in top-down order so that each node's parent is changed
at most once, as in Step~\ref{substep: MRP_Phase_1_delete_loop}.
\qed
\end{proof}

%%%%%%%%%%%%%%%%%%%%%%%%%%%%%%%%%%%%%%%%%%%%%%%%%%%%%%%%%%%%%%%%%%%%%%%%%%%%%%

\section{Constructing the frequency difference consensus tree}
\label{section: Frequency_difference_consensus_tree}

Here, we present an algorithm for finding the frequency consensus tree
of~$\SSS$ in $\min \{O(k n^{2})$, $O(k n (k + \log^{2}n))\}$ time.
It is called \texttt{Frequency\_Difference} and is described in
Section~\ref{subsection: Algorithm_Frequency_Difference} below.
The algorithm uses the procedure \texttt{Merge\_Trees} as well as
a new procedure named \texttt{Filter\_Clusters} whose details are given in
Section~\ref{subsection: Procedure_Filter_Clusters}.

For each tree~$T_j \in \SSS$ and each node~$u \in V(T_j)$, define
the \emph{weight of~$u$} as the value $|K_{\Lambda(T_j[u])}(\SSS)|$,
i.e., the number of trees from~$\SSS$ in which the cluster~$\Lambda(T_j[u])$
occurs, and denote it by~$w(u)$.
For convenience, also define $w(C) = w(u)$, where $C = \Lambda(T_j[u])$.
The input to Procedure~\texttt{Filter\_Clusters} is two trees~$T_A$, $T_B$
with $\Lambda(T_A) = \Lambda(T_B) = L$ such that every cluster occurring
in~$T_A$ or~$T_B$ also occurs in at least one tree in~$\SSS$,
and the output is a copy of~$T_A$ in which every cluster that is incompatible
with some cluster in~$T_B$ with a higher weight has been removed.
Formally, the output of~\texttt{Filter\_Clusters} is a tree~$T$ with
$\Lambda(T) = L$ such that
$\C(T) = \{\Lambda(T_A[u]) \,:\,
  u \in V(T_A) \textnormal{ and }
  w(u) > w(x) \textnormal{ for every } x \in V(T_B) \textnormal{ with }
  \Lambda(T_A[u]) \not\smile \Lambda(T_B[x])\}$.

%%%

\subsection{Algorithm~\texttt{Frequency\_Difference}}
\label{subsection: Algorithm_Frequency_Difference}

We first describe Algorithm~\texttt{Frequency\_Difference}.
Refer to Fig.~\ref{figure: Algorithm_Frequency_Difference} for the pseudocode.

\begin{figure}[t!]
\begin{center}
\fbox{
\begin{algorithm2}{\textnormal{\texttt{Frequency\_Difference}}}
\smallskip
\INPUT{
  A set $\SSS \,=\, \{T_1, T_2, \dots, T_k\}$ of trees with
  $\Lambda(T_1) = \Lambda(T_2) = \dots = \Lambda(T_k)$.
}
\OUTPUT{
  The frequency difference consensus tree of~$\SSS$.
}
\smallskip

\nullSTEP{
  /* Preprocessing */
}

\STEP{
\label{step: FD_compute_weights}
  Compute $w(C)$ for every cluster $C$ occurring in~$\SSS$.
}

\medskip

\nullSTEP{
  /* Main algorithm */
}

\STEP{
\label{step: FD_start_of_main_algorithm}
  $T := T_1$
}

\STEP{
\label{step: FD_build_forward_frequency_difference}
  \textbf{for} $j := 2$ \textbf{to} $k$ \textbf{do}
  \nullSTEP{\hspace*{5mm}
    $A := \texttt{Filter\_Clusters}(T, T_j)$;\,
    $B := \texttt{Filter\_Clusters}(T_j, T)$
  }
  \nullSTEP{\hspace*{5mm}
    $T := \texttt{Merge\_Trees}(A, B)$
  }
  \textbf{endfor}
}

\STEP{
\label{step: FD_remove_non-frequency_difference_clusters}
  \textbf{for} $j := 1$ \textbf{to} $k$ \textbf{do}
  $T := \texttt{Filter\_Clusters}(T, T_j)$
}

\STEP{
  \RETURN{$T$}
}
\end{algorithm2}
}
\caption{Algorithm~\texttt{Frequency\_Difference}
for constructing the frequency difference consensus tree.}
\label{figure: Algorithm_Frequency_Difference}
\end{center}
\end{figure}

The algorithm starts by computing the weight~$w(C)$ of every cluster~$C$
occurring in~$\SSS$ in a preprocessing step
(Step~\ref{step: FD_compute_weights}).
Next, let $\C(S)$ for any set~$S$ of trees denote the union
$\bigcup_{T_i \in S} \C(T_i)$,
and for any $j \in \{1,2,\dots,k\}$, define a
\emph{forward frequency difference consensus tree of~$\{T_1,T_2,\dots,T_j\}$}
as any tree that includes every cluster~$C$ in $\C(\{T_1,T_2,\dots,T_j\})$
satisfying $w(C) > w(X)$ for all $X \in \C(\{T_1,T_2,\dots,T_j\})$ with
$C \not\smile X$.
Steps~\ref{step: FD_start_of_main_algorithm}--\ref{step: FD_build_forward_frequency_difference}
use Procedure~\texttt{Filter\_Clusters} from
Section~\ref{subsection: Procedure_Filter_Clusters} to build a tree~$T$ that,
after any iteration $j \in \{1,2,\dots,k\}$, is a forward frequency difference
consensus tree of~$\{T_1,T_2,\dots,T_j\}$,
as proved in Lemma~\ref{lemma: forward_frequency_difference} below.
After iteration~$k$, $\C(T)$~contains all frequency difference clusters
of~$\SSS$ but possibly some other clusters as well, so
Step~\ref{step: FD_remove_non-frequency_difference_clusters} applies
\texttt{Filter\_Clusters} again to remove all non-frequency difference
clusters of~$\SSS$ from~$T$.

\begin{lemma}
\label{lemma: forward_frequency_difference}
For any $j \in \{2,3,\dots,k\}$, suppose that $T$ is a forward frequency
difference consensus tree of $\{T_1,T_2,\dots,T_{j-1}\}$.
%Then \textnormal{\texttt{Merge\_Trees}}$(A, B)$, where
%$A :=$ \textnormal{\texttt{Filter\_Clusters}}$(T, T_j)$ and
%$B :=$ \textnormal{\texttt{Filter\_Clusters}}$(T_j, T)$,
%is a forward frequency difference consensus tree of $\{T_1,T_2,\dots,T_j\}$.
Let $A :=$ \textnormal{\texttt{Filter\_Clusters}}$(T, T_j)$ and
$B :=$ \textnormal{\texttt{Filter\_Clusters}}$(T_j, T)$.
Then \textnormal{\texttt{Merge\_Trees}}$(A, B)$
is a forward frequency difference consensus tree of $\{T_1,T_2,\dots,T_j\}$.
\end{lemma}

%%% END OF COMMENT

\begin{proof}
(Omitted from the conference version due to space constraints.)
\qed
\end{proof}

\begin{theorem}
\label{theorem: Frequency_Difference}
Algorithm~\textnormal{\texttt{Frequency\_Difference}}
constructs the frequency difference consensus tree of~$\SSS$ in
%$\min \{O(k n^{2}),\, O(k n (k + \log^{2}n))\}$ time.
$\min \{O(k n^{2}),\, O(k^{2} n)\} \,+\, O(k \cdot f(n))$ time, where $f(n)$
is the running time of Procedure~\textnormal{\texttt{Filter\_Clusters}}.
\end{theorem}

\begin{proof}
After completing iteration~$k$ of
Step~\ref{step: FD_build_forward_frequency_difference},
$\C(T)$~is a superset of the set of all frequency difference clusters
of~$\SSS$ by Lemma~\ref{lemma: forward_frequency_difference}.
Next, Step~\ref{step: FD_remove_non-frequency_difference_clusters} removes
all non-frequency difference clusters of~$\SSS$, so the output will be
the frequency difference consensus tree of~$\SSS$.

%\medskip

To analyze the time complexity, first consider how to compute all the weights
in Step~\ref{step: FD_compute_weights}.
One method is to first fix an arbitrary ordering of~$L$ and represent every
cluster~$C$ of~$L$ as a bit vector of length~$n$
(for every $i \in \{1,2,\dots,n\}$, the $i$th bit is set to~$1$ if and only if
the $i$th leaf label belongs to~$C$).
Then, spend $O(k n^{2})$ time to construct a list of bit vectors for all
$O(k n)$ clusters occurring in~$\SSS$ by a bottom-up traversal of each tree
in~$\SSS$, sort the resulting list of bit vectors by radix sort, and traverse
the sorted list to identify the number of occurrences of each cluster.
All this takes $O(k n^{2})$ time.
An alternative method, which uses $O(k^{2} n)$ time, is to initialize
the weight of every node in~$\SSS$ to~$1$ and then,
for $j \in \{1,2,\dots,k\}$, apply Day's algorithm
(see Lemma~\ref{lemma: Day's_algorithm})
with $T_{\mathit{ref}} = T_j$ and $T$ ranging over all $T_i$ with
$1 \leq i \leq k,\, i \neq j$ to find all clusters in~$T$ that also occur
in~$T_j$ and increase the weights of their nodes in~$T$ by~$1$.
Therefore, Step~\ref{step: FD_compute_weights} takes
$\min \{O(k n^{2}),\, O(k^{2} n)\}$ time.
Next, Steps~\ref{step: FD_build_forward_frequency_difference}
and~\ref{step: FD_remove_non-frequency_difference_clusters}
make $O(k)$ calls to the procedures
\texttt{Merge\_Trees} and \texttt{Filter\_Clusters}.
The running time of \texttt{Merge\_Trees} is $O(n)$ by
Lemma~\ref{lemma: Procedure_Merge_Trees} and the running time of
\texttt{Filter\_Clusters} is $f(n) = \Omega(n)$,
so Steps~\ref{step: FD_build_forward_frequency_difference}
and~\ref{step: FD_remove_non-frequency_difference_clusters}
take $O(k \cdot f(n))$ time.
\qed
\end{proof}

Lemma~\ref{lemma: Procedure_Filter_Clusters} in the next
subsection shows that $f(n) = O(n \log^{2} n)$ is possible, which yields:

\begin{corollary}
\label{corollary: Frequency_Difference}
Algorithm~\textnormal{\texttt{Frequency\_Difference}}
constructs the frequency difference consensus tree of~$\SSS$ in
$\min \{O(k n^{2}),\, O(k n (k + \log^{2}n))\}$ time.
\end{corollary}

%%%

\subsection{Procedure~\texttt{Filter\_Clusters}}
\label{subsection: Procedure_Filter_Clusters}

Recall that for any node~$u$ in any input tree~$T_j$, its weight~$w(u)$
%equals
is
$|K_{\Lambda(T_j[u])}(\SSS)|$.
Also, $w(C) = w(u)$, where $C = \Lambda(T_j[u])$.
We assume that all $w(u)$-values have been computed in a preprocessing step
and are available.

Let $T$ be a tree.
For every nonempty $X \subseteq V(T)$, $\mathit{lca}^{T}(X)$~denotes
the lowest common ancestor of~$X$ in~$T$.
To obtain a fast solution for \texttt{Filter\_Clusters}, we need the next
lemma.
\begin{lemma}
\label{lemma: compatible_lca_path}
Let $T$ be a tree, let $X$ be any cluster of~$\Lambda(T)$, and let
$r_X = \mathit{lca}^{T}(X)$.
For any $v \in V(T)$, it holds that $X \not\smile \Lambda(T[v])$
if and only if:
(1)~$v$ lies on a path from a child of~$r_X$ to some leaf belonging to~$X$;
and (2)~$\Lambda(T[v]) \not\subseteq X$.
\end{lemma}

\begin{proof}
Given $T$, $X$, $r_X$, and $v$ as in the lemma statement, there are four
possible cases:
(i)~$v$ is a proper ancestor of~$r_X$ or equal to~$r_X$;
(ii)~$v$ lies on a path from a child of~$r_X$ to some leaf in~$X$ and
all leaf descendants of~$v$ belong to~$X$;
(iii)~$v$ lies on a path from a child of~$r_X$ to some leaf in~$X$ and
not all leaf descendants of~$v$ belong to~$X$;
or (iv)~$v$ is a proper descendant of~$r_X$ that does not lie on any path
from a leaf in~$X$ to~$r_X$.
In case~(i), $X \subseteq \Lambda(T[v])$.
In case~(ii), $\Lambda(T[v]) \subseteq X$.
In case~(iii), $\Lambda(T[v]) \not\subseteq X$ while
$X \cap \Lambda(T[v]) \neq \emptyset$.
In case~(iv), $X \cap \Lambda(T[v]) = \emptyset$.
By the definition of compatible clusters, $X \not\smile \Lambda(T[v])$
if and only if case~(iii) occurs.
\qed
\end{proof}

Lemma~\ref{lemma: compatible_lca_path} leads to an $O(n^{2})$-time method
for \texttt{Filter\_Clusters}, which we now briefly describe.
For each node $u \in V(T_A)$ in top-down order, do the following:
Let $X = \Lambda(T_A[u])$ and find all $v \in V(T_B)$ such that
$X \not\smile \Lambda(T_B[v])$ in $O(n)$ time
by doing bottom-up traversals of~$T_B$ to first mark all ancestors of leaves
belonging to~$X$ that are proper descendants of the lowest common ancestor
of~$X$ in~$T_B$, and then unmarking all marked nodes that have no leaf
descendants outside of~$X$.
By Lemma~\ref{lemma: compatible_lca_path}, $X \not\smile \Lambda(T_B[v])$
if and only if $v$ is one of the resulting marked nodes.
If $w(u) \leq w(v)$ for any such~$v$ then do a $\mathit{delete}$ operation
on~$u$ in~$T_A$.
Clearly, the total running time is $O(n^{2})$.
(This simple method gives $f(n) = O(n^{2})$ in
Theorem~\ref{theorem: Frequency_Difference}
in Section~\ref{subsection: Algorithm_Frequency_Difference}, and hence
a total running time of $O(k n^{2})$ for
Algorithm~\texttt{Frequency\_Difference}.)
Below, we refine this idea to get an even faster solution for
\texttt{Filter\_Clusters}.

\medskip
\medskip

\noindent
\textbf{High-level description:}
We use the \emph{centroid path decomposition} technique~\cite{CF-CHPT00}
to divide the nodes of~$T_A$ into a so-called centroid path and a set of
side trees.
A \emph{centroid path} of~$T_A$ is defined as a path in~$T_A$ of the form
$\pi = \langle p_{\alpha},p_{\alpha-1},\dots,p_{1} \rangle$, where
$p_{\alpha}$ is the root of~$T_A$, the node~$p_{i-1}$ for every
$i \in \{2,\dots,\alpha\}$ is any child of~$p_{i}$ with the maximum number
of leaf descendants, and $p_{1}$~is a leaf.
Given a centroid path~$\pi$, removing $\pi$ and all its incident edges
from~$T_A$ produces a set~$\sigma(\pi)$ of disjoint trees whose root nodes
are children of nodes belonging to~$\pi$ in~$T_A$; these trees are called
the \emph{side trees} of~$\pi$.
Importantly, $|\Lambda(\tau)| \leq n/2$ for every side tree $\tau$ of~$\pi$.
Also, $\{\Lambda(\tau) : \tau \in \sigma(\pi)\}$ forms a partition
of~$L \setminus \{p_{1}\}$.
Furthermore, if $\pi$ is a centroid path of~$T_A$ then
the cluster collection~$\C(T_A)$ can be written recursively as
$\C(T_A) \,=\,
  \bigcup_{\tau \in \sigma(\pi)} \C(\tau)
  \,\cup\,
  \bigcup_{p_{i} \in \pi} \{\Lambda(T_A[p_{i}])\}
$.
Intuitively, this allows the cluster collection of~$T_A$ to be broken into
smaller sets that can be checked more easily, and then put together again at
the end.

The fast version of \texttt{Filter\_Clusters} is shown in
Fig.~\ref{figure: Algorithm_Filter_Clusters}.
It first computes
a centroid path~$\pi = \langle p_{\alpha},p_{\alpha-1},\dots,p_{1} \rangle$
of~$T_A$ and the set~$\sigma(\pi)$ of side trees of~$\pi$ in
Step~\ref{step: FC_compute_centroid_path}.
Then, in
Steps~\ref{step: FC_side_trees_init}--\ref{step: FC_side_trees_loop},
it applies itself recursively to each side tree of~$\pi$ to get rid of any
cluster in $\bigcup_{\tau \in \sigma(\pi)} \C(\tau)$ that is incompatible
with some cluster in~$T_B$ with a higher weight than itself, and
the remaining clusters are inserted into a temporary tree~$R_{s}$.
Next,
Steps~\ref{step: FC_centroid_path_init}--\ref{step: FC_centroid_path_loop}
check all clusters in $\bigcup_{p_{i} \in \pi} \{\Lambda(T_A[p_{i}])\}$ to
determine which of them are not incompatible with any cluster in~$T_B$
with a higher weight, and create a temporary tree~$R_{c}$ whose cluster
collection consists of all those clusters that pass this test.
Finally, Step~\ref{step: FC_Merge_Trees} combines the cluster collections
of~$R_{s}$ and~$R_{c}$ by applying the procedure \texttt{Merge\_Trees}.
The details of Procedure~\texttt{Filter\_Clusters} are discussed next.

\begin{figure}
\begin{center}
\fbox{
\begin{algorithm2}{\textnormal{\texttt{Filter\_Clusters}}}
%\smallskip
\INPUT{
Two trees~$T_A$, $T_B$ with $\Lambda(T_A) = \Lambda(T_B) = L$ such that
every cluster occurring in~$T_A$ or~$T_B$ also occurs in at least one tree
in~$\SSS$.
}
%\smallskip
\OUTPUT{
A tree~$T$ with $\Lambda(T) = L$ such that
$\C(T) = \{\Lambda(T_A[u]) \,:\,
  u \in V(T_A)$
  and
  $w(u) > w(x) \textnormal{ for every } x \in V(T_B) \textnormal{ with }
  \Lambda(T_A[u]) \not\smile \Lambda(T_B[x])\}$.
}

%\smallskip

\STEP{
\label{step: FC_compute_centroid_path}
  Compute a centroid path
  $\pi = \langle p_{\alpha},p_{\alpha-1},\dots,p_{1} \rangle$
  of~$T_A$, where $p_{\alpha}$ is the root of~$T_A$ and $p_{1}$ is a leaf,
  and compute the set~$\sigma(\pi)$ of side trees of~$\pi$.
}

\medskip

\nullSTEP{
  /* Handle the side trees. */
}

\STEP{
\label{step: FC_side_trees_init}
  Let $R_{s}$ be a tree consisting only of a root node and a single leaf
  labeled by~$p_{1}$.
}

\STEP{
\label{step: FC_side_trees_loop}
  \textbf{for} each side tree $\tau \in \sigma(\pi)$ \textbf{do}
  \nullSTEP{\hspace*{5mm}
    $\tau' := \texttt{Filter\_Clusters}(\tau,\, T_B || \Lambda(\tau))$
  }
  \nullSTEP{\hspace*{5mm}
    Attach the root of $\tau'$ to the root of~$R_{s}$.
  }
  \textbf{endfor}
}

\medskip

\nullSTEP{
  /* Handle the centroid path. */
}

\STEP{
\label{step: FC_centroid_path_init}
  Let $R_{c}$ be a tree with $\Lambda(R_{c}) = L$ where every leaf is
  directly attached to the root.
%}
%
%\STEP{
  Let $BT$ be an empty binary search tree.
  For every
%node
  $x \in V(T_B)$, initialize $\mathit{counter}(x) := 0$.
  Do a bottom-up traversal of~$T_B$ to precompute $|\Lambda(T_B[x])|$
  for every $x \in V(T_B)$.
%}
%
%\STEP{
  Preprocess $T_B$ for answering $\mathit{lca}$-queries.
%}
%
%\STEP{
  Let $\beta_{1} := 0$.
}

\STEP{
\label{step: FC_centroid_path_loop}
  \textbf{for} $i := 2$ \textbf{to} $\alpha$ \textbf{do}
  \subSTEP{\hspace*{5mm}
    Let $D$ be the set of leaves in
    $\Lambda(T_A[p_{i}]) \setminus \Lambda(T_A[p_{i-1}])$.
  }
  \subSTEP{\hspace*{5mm}
    Compute $r_{i} := \mathit{lca}^{T_B}(\{r_{i-1}\} \cup D)$.
%    \hspace*{3mm}
    \hspace*{2mm}
%    /* $r_{i}$ is now equal to $\mathit{lca}^{T_B}(\Lambda(T_A[p_{i}]).$ */
    /* $r_{i}$ now equals $\mathit{lca}^{T_B}(\Lambda(T_A[p_{i}]).$ */
  }
  \subSTEP{\hspace*{5mm}
\label{substep: FC_insert_centroid_subpath_into_BT}
    Insert every node belonging to the path from~$r_{i}$ to~$r_{i-1}$,
    except~$r_{i}$, into~$BT$.
  }
  \subSTEP{\hspace*{5mm}
\label{substep: FC_insert_rest_into_BT}
    \textbf{for} each $x \in D$ \textbf{do}
      \nullSTEP{\hspace*{10mm}
        Insert $x$ into $BT$.
      }
%      \nullSTEP{\hspace*{10mm}
%        \textbf{endwhile}
%      }
%      \nullSTEP{\hspace*{10mm}
%        \textbf{while} ($\mathit{parent}(x)$ is not in $BT$ and
%                        $\mathit{parent}(x) \neq r_i$)
%        \textbf{do}
%        \nullSTEP{\hspace*{15mm}
%          $x := \mathit{parent(x)}$
%        }
%        \nullSTEP{\hspace*{15mm}
%          Insert $x$ into $BT$.
%        }
%      \nullSTEP{\hspace*{10mm}
%        \textbf{endwhile}
%      }
      \nullSTEP{\hspace*{10mm}
        \textbf{while} ($\mathit{parent}(x)$ is not in $BT$ and
                        $\mathit{parent}(x) \neq r_i$)
        \textbf{do}
        \nullSTEP{\hspace*{15mm}
          $x := \mathit{parent(x)}$;
          insert $x$ into $BT$.
        }
    \nullSTEP{\hspace*{5mm}
      \textbf{endfor}
    }
    }
  }
  \subSTEP{\hspace*{5mm}
\label{substep: FC_update_BT}
    \textbf{for} each $x \in D$ \textbf{do}
      \nullSTEP{\hspace*{10mm}
        $\mathit{counter}(x) := \mathit{counter}(x) + 1$
      }
%      \nullSTEP{\hspace*{10mm}
%        \textbf{while} ($\mathit{counter}(x) = |\Lambda(T_B[x])|$)
%        \textbf{do}
%        \nullSTEP{\hspace*{15mm}
%          Remove $x$ from $BT$.
%        }
%        \nullSTEP{\hspace*{15mm}
%          $\mathit{counter}(\mathit{parent}(x)) :=
%            \mathit{counter}(\mathit{parent}(x)) + |\Lambda(T_B[x])|$
%        }
%        \nullSTEP{\hspace*{15mm}
%          $x := \mathit{parent}(x)$
%        }
%      \nullSTEP{\hspace*{10mm}
%        \textbf{endwhile}
%      }
      \nullSTEP{\hspace*{10mm}
        \textbf{while} ($\mathit{counter}(x) = |\Lambda(T_B[x])|$)
        \textbf{do}
        \nullSTEP{\hspace*{15mm}
          $\mathit{counter}(\mathit{parent}(x)) :=
            \mathit{counter}(\mathit{parent}(x)) + |\Lambda(T_B[x])|$
        }
        \nullSTEP{\hspace*{15mm}
          Remove $x$ from $BT$;
          $x := \mathit{parent}(x)$
        }
      \nullSTEP{\hspace*{10mm}
        \textbf{endwhile}
      }
    \nullSTEP{\hspace*{5mm}
      \textbf{endfor}
    }
    }
  }
  \subSTEP{\hspace*{5mm}
\label{substep: FC_retrieve_maximum_weight}
%    Let $M$ be the maximum weight of all nodes in~$BT$;
    Let $M := $ maximum weight of a node in~$BT$;
    \textbf{if} $BT$ is empty \textbf{then} $M := 0$.
  }
  \subSTEP{\hspace*{5mm}
\label{substep: FC_special_nodes}
    Compute $\beta_{i}$, and \textbf{if} $\beta_{i} > M$ \textbf{then}
    let $M := \beta_{i}$.
  }
  \subSTEP{\hspace*{5mm}
\label{substep: FC_centroid_path_insert}
    \textbf{if} ($w(\Lambda(T_A[p_{i}])) > M$) \textbf{then}
%    do an $\mathit{insert}$ operation to include $\Lambda(T_A[p_{i}])$
%    in~$R_{c}$.
    put $\Lambda(T_A[p_{i}])$ in~$R_{c}$ by an $\mathit{insert}$ operation.
  }
  \textbf{endfor}
}

\medskip

\nullSTEP{
  /* Combine the surviving clusters. */
}

\STEP{
\label{step: FC_Merge_Trees}
  $T := \texttt{Merge\_Trees}(R_{s}, R_{c})$
}

\STEP{
  \RETURN{$T$}
}
\end{algorithm2}
}
\caption{The procedure~\texttt{Filter\_Clusters}.}
\label{figure: Algorithm_Filter_Clusters}
\end{center}
\end{figure}

\medskip

\noindent
\textbf{Steps~\ref{step: FC_side_trees_init}--\ref{step: FC_side_trees_loop}
(handling the side trees):}
For every nonempty $C \subseteq \Lambda(T)$, define $T | C$
(``the subtree of~$T$ induced by~$C$''; see, e.g.,~\cite{CF-CHPT00})
as the tree~$T'$ with leaf label set~$C$ and
internal node set~$\{\mathit{lca}^{T}(\{a,b\}) : a,b \in C\}$
which preserves the ancestor relations from~$T$, i.e., which satisfies
$\mathit{lca}^{T}(C') = \mathit{lca}^{T'}(C')$ for all nonempty
$C' \subseteq C$.
Now, let $\sigma(\pi)$ be the set of side trees of the centroid path~$\pi$
computed in Step~\ref{step: FC_compute_centroid_path}.
For each $\tau \in \sigma(\pi)$, define a weighted tree $T_B || \Lambda(\tau)$
as follows.
First, construct $T_B | \Lambda(\tau)$ and let the weight of each node in
this tree equal its weight in~$T_B$.
Next, for each edge $(u,v)$ in~$T_B | \Lambda(\tau)$, let $P$ be the path
in~$T_B$ between~$u$ and~$v$, excluding~$u$ and~$v$;
if $P$ is not empty then create a new node~$z$ in $T_B | \Lambda(\tau)$,
replace the edge~$(u,v)$ by the two edges $(u,z)$ and~$(z,v)$, and set
the weight of~$z$ to the maximum weight of all nodes belonging to~$P$.
Each such~$z$ is called a \emph{special node} and has exactly one child.
See Fig.~\ref{figure: create_restricted_weighted_tree} for an example.
We extend the concept of ``compatible'' to special nodes as follows:
if $C \subseteq L$ and $z$ is a special node in~$T_B || \Lambda(\tau)$ then
$C \smile z$ if and only if
$C$~and~$\Lambda((T_B || \Lambda(\tau))[z])$ are disjoint or
$(T_B || \Lambda(\tau))[z]$ has no proper descendant that is a special node.
The obtained tree $T_B || \Lambda(\tau)$ satisfies
$\Lambda(\tau) = \Lambda(T_B || \Lambda(\tau))$ and has the property that
for every cluster~$C$ in~$\C(\tau)$,
$\max \{w(X) : X \in \C(T_B)
  \textnormal{ and } C \not\smile X\}$
is equal to
$\max \{w(X) : X \in \C(T_B || \Lambda(\tau))
  \textnormal{ and } C \not\smile X\}$.

\begin{figure}[t!]
\begin{center}
  \includegraphics[scale=0.35]{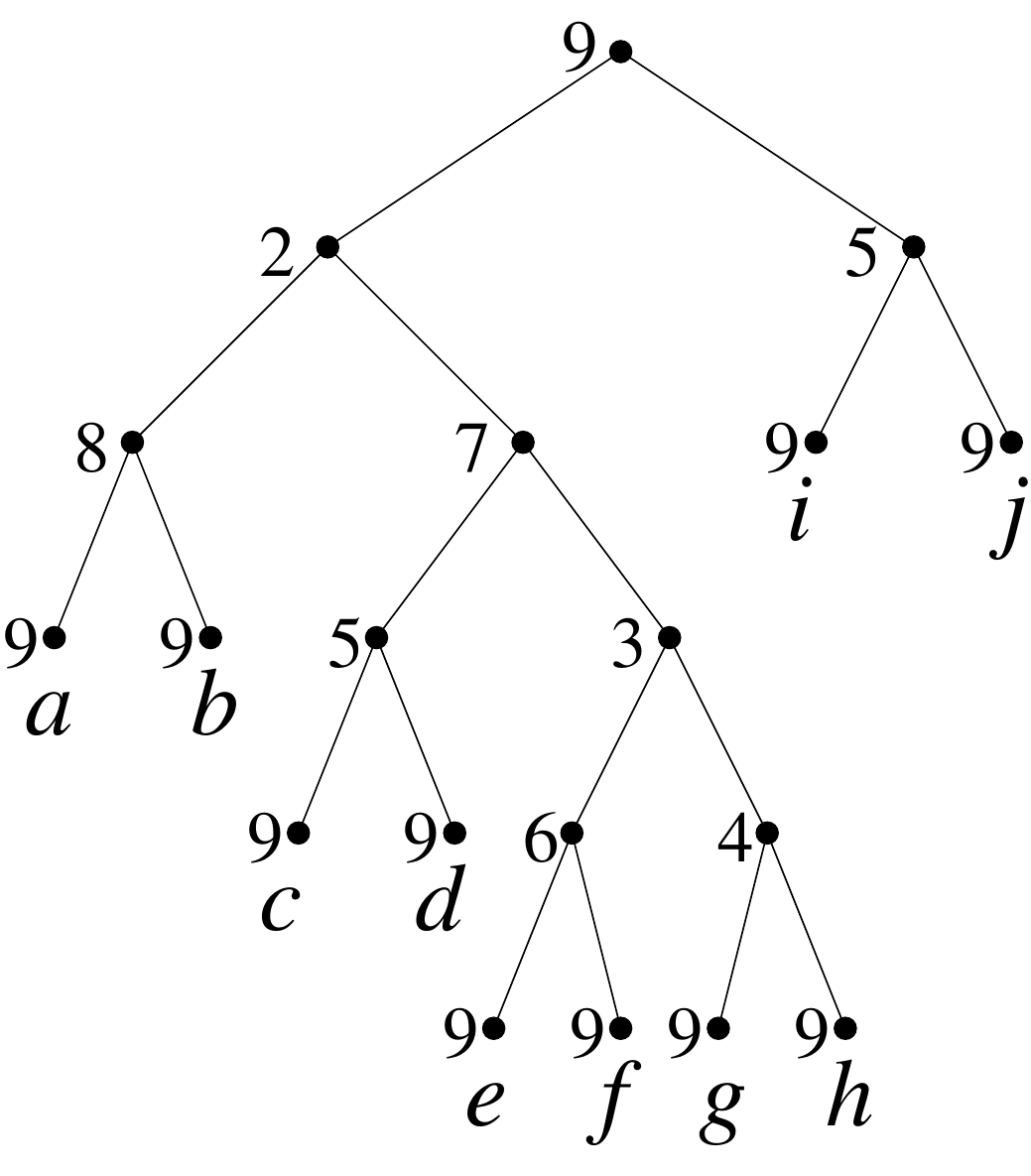}
%\hspace*{20mm}
\hspace*{15mm}
  \includegraphics[scale=0.35]{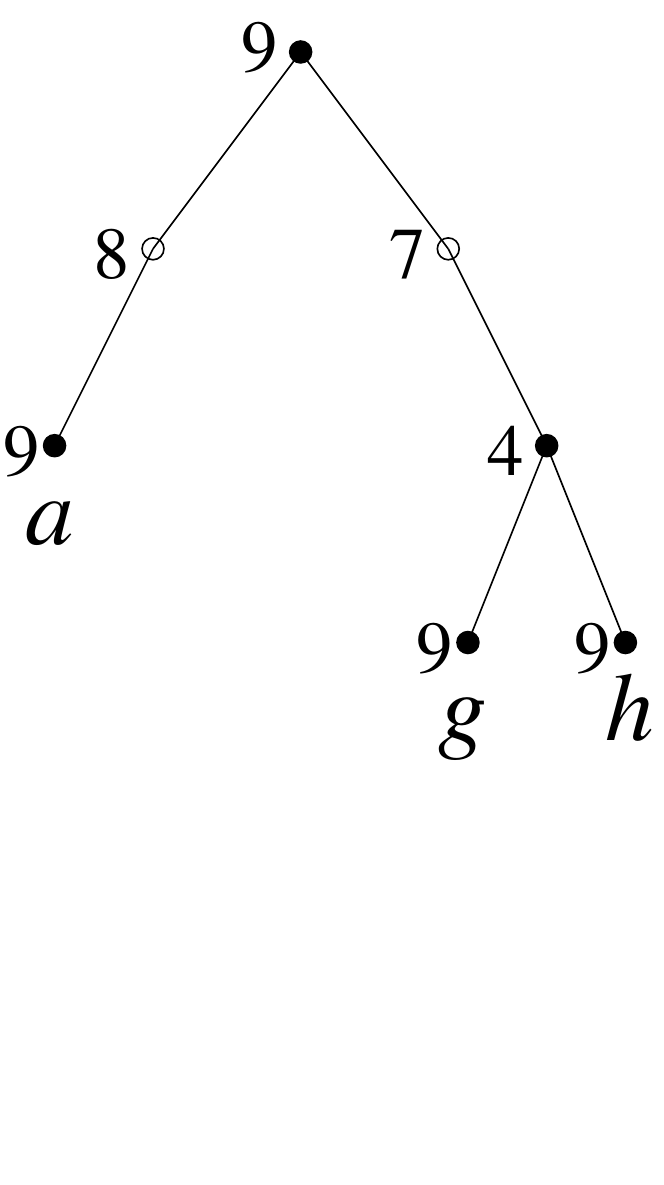}
\caption{Let $T_B$ be the tree (with node weights) on the left.
The tree~$T_B || \{a,g,h\}$ is shown on the right.
}
\label{figure: create_restricted_weighted_tree}
\end{center}
\end{figure}

After constructing $T_B || \Lambda(\tau)$, \texttt{Filter\_Clusters}
%can be
is
applied to $(\tau,\, T_B || \Lambda(\tau))$
recursively to remove all bad clusters from~$\tau$.
For each $\tau \in \sigma(\pi)$, the resulting tree is denoted by~$\tau'$.
All the clusters of~$\tau'$ are inserted into the tree~$R_{s}$ by directly
attaching $\tau'$ to the root of~$R_{s}$.
Since $\{\Lambda(\tau') : \tau \in \sigma(\pi)\}$ forms a partition
of~$L \setminus \{p_{1}\}$, every leaf label in~$L$ appears exactly once
in~$R_{s}$ and we have
$\C(R_{s}) =
  \{\Lambda(T_A[u]) \,:\,
  u \in V(\tau) \textnormal{ for some } \tau \in \sigma(\pi)$
  and
  $w(u) > w(x) \textnormal{ for every } x \in V(T_B) \textnormal{ with }
  \Lambda(T_A[u]) \not\smile \Lambda(T_B[x])\}
  \,\cup\, \{L\}$
after Step~\ref{step: FC_side_trees_loop} is finished.

\medskip

\noindent
\textbf{Steps~\ref{step: FC_centroid_path_init}--\ref{step: FC_centroid_path_loop}
(handling the centroid path):}
The clusters $\bigcup_{p_{i} \in \pi} \{\Lambda(T_A[p_{i}])\}$ on
the centroid path are nested
because $p_{i}$ is the parent of~$p_{i-1}$, so
$\Lambda(T_A[p_{i-1}]) \subseteq \Lambda(T_A[p_{i}])$ for every
$i \in \{2,3,\dots,\alpha\}$.
The main loop (Step~\ref{step: FC_centroid_path_loop}) checks each of these
clusters in order of increasing cardinality.

The algorithm maintains a binary search tree~$BT$ that, right after
Step~\ref{substep: FC_update_BT} in any iteration~$i$ of the main loop is
complete, contains every node~$x$ from~$T_B$ with
$\Lambda(T_A[p_{i}]) \not\smile \Lambda(T_B[x])$.
Whenever a node~$x$ is inserted into~$BT$, its key is set to
the weight~$w(T_B[x])$.
Using $BT$, Step~\ref{substep: FC_retrieve_maximum_weight} retrieves
the weight~$M$ of the heaviest cluster in~$T_B$ that is incompatible
with~$\Lambda(T_A[p_{i}])$ (if any).
Then, Step~\ref{substep: FC_special_nodes} computes a value~$\beta_{i}$,
defined as the maximum weight of all special nodes in~$T_B$ (if any) that are
incompatible with the current~$T_A[p_{i}]$; if $\beta_{i} > M$ then $M$ is
set to~$\beta_{i}$.
Step~\ref{substep: FC_centroid_path_insert} saves $\Lambda(T_A[p_{i}])$ by
inserting it into the tree~$R_{c}$ if its weight is strictly greater
than~$M$.
After Step~\ref{step: FC_centroid_path_loop} is done,
$\C(R_{c}) =
  \{\Lambda(T_A[u]) \,:\, u \in \pi \textnormal{ and }
  w(u) > w(x) \textnormal{ for every } x \in V(T_B) \textnormal{ with }
  \Lambda(T_A[u]) \not\smile \Lambda(T_B[x])\}$.

In order to update~$BT$ correctly while moving upwards along~$\pi$ in
Step~\ref{step: FC_centroid_path_loop}, the algorithm relies on
Lemma~\ref{lemma: compatible_lca_path}.
In each iteration $i \in \{2,3,\dots,\alpha\}$ of
Step~\ref{step: FC_centroid_path_loop}, $r_{i}$~is the lowest common ancestor
in~$T_B$ of~$\Lambda(T_A[p_{i}])$.
By Lemma~\ref{lemma: compatible_lca_path}, the clusters in~$T_B$ that are
incompatible with $\Lambda(T_A[p_{i}])$ are of the form~$T_B[v]$ where:
(1)~$v$ lies on a path in~$T_B$ from a child of~$r_{i}$ to a leaf
in~$\Lambda(T_A[p_{i}])$;
and (2)~$\Lambda(T[v]) \not\subseteq \Lambda(T_A[p_{i}])$.
Accordingly, $BT$~is updated in
Steps~\ref{substep: FC_insert_centroid_subpath_into_BT}--\ref{substep: FC_update_BT}
as follows.
Condition~(1) is taken care of by first inserting all nodes from~$T_B$
between~$r_{i-1}$ and~$r_{i}$ except~$r_i$ into~$BT$ in
Step~\ref{substep: FC_insert_centroid_subpath_into_BT}
and then inserting all leaf descendants of~$p_{i}$ that are not descendants
of~$p_{i-1}$, along with any of their ancestors in~$T_B$ that were not already
in~$BT$, into~$BT$ in Step~\ref{substep: FC_insert_rest_into_BT}.
Finally, Step~\ref{substep: FC_update_BT} enforces condition~(2) by using
counters to locate and remove all nodes from~$BT$ (if any) whose clusters are
proper subsets of~$\Lambda(T_A[p_{i}])$.
To do this, $\mathit{counter}(x)$ for every node~$x$ in~$T_B$ is updated so
that it stores the number of leaves
in~$\Lambda(T_B[x]) \cap \Lambda(T_A[p_{i}])$ for the current~$i$,
and if $\mathit{counter}(x)$ reaches the value $|\Lambda(T_B[x])|$ then $x$
is removed from~$BT$.

To compute~$\beta_{i}$ in Step~\ref{substep: FC_special_nodes}, take
the maximum of:
(i)~$\beta_{i-1}$;
(ii)~the weights of all special nodes on the path between~$r_{i}$
and~$r_{i-1}$ in~$T_B$; and
(iii)~the weights of all special nodes that belong to a path between~$r_{i}$
and a leaf in~$D$.

%%% END OF COMMENT

\begin{lemma}
\label{lemma: Procedure_Filter_Clusters}
Procedure~\textnormal{\texttt{Filter\_Clusters}} runs in
$O(n \log^{2} n)$ time.
\end{lemma}

\begin{proof}
(Omitted from the conference version due to space constraints.)
\qed
\end{proof}

%%%%%%%%%%%%%%%%%%%%%%%%%%%%%%%%%%%%%%%%%%%%%%%%%%%%%%%%%%%%%%%%%%%%%%%%%%%%%%

\section{Implementations}
\label{section: Implementations}

As noted in Section~\ref{subsection: Previous_work}, there does not seem to
be any publicly available implementation for the majority rule~(+) consensus
tree.
To fill this void, we implemented Algorithm \texttt{Maj\_Rule\_Plus} from
Section~\ref{section: Majority_rule_plus_consensus_tree} in~C++ and included
it in the source code of the FACT (Fast Algorithms for Consensus Trees)
package~\cite{JSS_13a} at:

%\medskip
\verb=http://compbio.ddns.comp.nus.edu.sg/~consensus.tree/=
%\medskip

\noindent
To test the implementation, we repeatedly applied it to 10~random sets of
trees for various specified values of~$(k,n)$, generated with the method
described in Section~6.2 of~\cite{JSS_13a}.
The following worst-case running times (in seconds) were obtained using
Ubuntu Nutty Narwhal, a 64-bit operating system with 8.00~GB RAM, and
a 2.20~GHz CPU:

%\medskip
\smallskip

\noindent
%\scalebox{0.878}{
\scalebox{0.715}{
\begin{tabular}{||c||c|c|c|c|c|c|c|c|c||}
\hline
\hline
$(k,n)$ &
$(100,500)$ &
$(100,1000)$ &
$(100,2000)$ &
$(100,5000)$ &
$(500,100)$ &
$(1000,100)$ &
$(2000,100)$ &
$(5000,100)$ &
$(1000,2000)$
\\ \hline\hline
Time &
0.63 &
1.51 &
2.99 &
6.78 &
0.65 &
1.29 &
2.72 &
6.66 &
27.29
\\ \hline\hline
\end{tabular}
}

%\bigskip
\medskip

The situation for the frequency difference consensus tree is less critical
as there already exist implementations, e.g., in the software package
TNT~\cite{software:GFN08}.
Nevertheless, it could be useful to implement our algorithm
\texttt{Frequency\_Difference} from
Section~\ref{section: Frequency_difference_consensus_tree} in the future and
compare its practical performance to TNT.
Before doing that, one should try to simplify the procedure
\texttt{Filter\_Clusters}.

%%%%%%%%%%%%%%%%%%%%%%%%%%%%%%%%%%%%%%%%%%%%%%%%%%%%%%%%%%%%%%%%%%%%%%%%%%%%%%
%%%%%%%%%%%%%%%%%%%%%%%%%%%%%%%%%%%%%%%%%%%%%%%%%%%%%%%%%%%%%%%%%%%%%%%%%%%%%%

\bibliography{Bibl_frequency_difference_c_t1_WABI_2013}
\bibliographystyle{plain}

%%%%%%%%%%%%%%%%%%%%%%%%%%%%%%%%%%%%%%%%%%%%%%%%%%%%%%%%%%%%%%%%%%%%%%%%%%%%%%
%%%%%%%%%%%%%%%%%%%%%%%%%%%%%%%%%%%%%%%%%%%%%%%%%%%%%%%%%%%%%%%%%%%%%%%%%%%%%%

\end{document}